\documentclass[a4paper,
               DIV=10, 
               titlepage=off]{scrartcl}

\usepackage[utf8]{inputenc}
\usepackage[english]{babel}
\usepackage{csquotes}
\usepackage{todonotes}
\usepackage{ytableau}
\usepackage[left=2cm,right=3cm,top=2cm,bottom=3cm]{geometry}

\usepackage{amsthm,amsmath,amsfonts,amssymb}
\usepackage{braket,bbm,nicefrac}
\usepackage{esvect}

\usetikzlibrary{arrows,arrows.meta,positioning}

\usepackage[colorlinks=true,linkcolor=blue,urlcolor=blue,citecolor=blue,anchorcolor=green,pdfusetitle,breaklinks]{hyperref}
\usepackage{dsfont,xspace,xparse}
\usepackage{lmodern,newtxtt,microtype}
\usepackage{tabularx,xifthen,afterpage,booktabs}
\usepackage{algpseudocode,algorithm}

\numberwithin{equation}{section}

\usepackage{microtype,parskip} 

\usepackage{upgreek}

\usepackage{cleveref}
\crefname{lemma}{lemma}{lemmas}
\crefname{proposition}{proposition}{propositions}
\crefname{definition}{definition}{definitions}
\crefname{theorem}{theorem}{theorems}
\crefname{conjecture}{conjecture}{conjectures}
\crefname{corollary}{corollary}{corollaries}
\crefname{example}{example}{examples}
\crefname{section}{section}{sections}
\crefname{appendix}{appendix}{appendices}
\crefname{figure}{fig.}{figs.}
\crefname{equation}{eq.}{eqs.}
\crefname{table}{table}{tables}
\crefname{item}{property}{properties}
\crefname{remark}{remark}{remarks}
\crefname{problem}{}{}

\newtheorem{theorem}{Theorem}

\newtheorem{lemma}[theorem]{Lemma}

\usepackage{listings}
\usepackage[backend=biber,
            style=trad-alpha,
            doi=false,
            isbn=false,
            url=false,
            maxbibnames=20,
            maxcitenames=2]{biblatex}
\bibliography{refs}
\newbibmacro{string+doi}[1]{\iffieldundef{doi}{#1}{\href{https://dx.doi.org/\thefield{doi}}{#1}}}
\DeclareFieldFormat{title}{\usebibmacro{string+doi}{\mkbibemph{#1}}}
\DeclareFieldFormat[article]{title}{\usebibmacro{string+doi}{\mkbibquote{#1}}}
\DeclareFieldFormat[incollection]{title}{\usebibmacro{string+doi}{\mkbibquote{#1}}}
\DeclareFieldFormat[inproceedings]{title}{\usebibmacro{string+doi}{\mkbibquote{#1}}}
\renewbibmacro{in:}{}

\newcommand{\C}{{\mathbb{C}}} 
\newcommand{\IC}{\C} 

\newcommand{\IN}{{\mathbb{N}}}

\newcommand{\qid}{{\mathbbm{1}}} 
\newcommand{\G}{\mathsf{G}} 
\renewcommand{\H}{\mathcal{H}} 

\DeclareMathOperator{\Id}{Id}
\DeclareMathOperator{\idd}{id}
\DeclareMathOperator{\im}{im}
\DeclareMathOperator{\GL}{GL}

\DeclareMathOperator{\tr}{tr}
\DeclareMathOperator{\Sym}{Sym}

\DeclareMathOperator{\ii}{\mathrm{i}}
\DeclareMathOperator{\ee}{\mathrm{e}}
\DeclareMathOperator{\poly}{poly}

\newcommand{\Norm}[1]{\left\|{#1}\right\|}

\newcommand{\ketbra}[2]{\ket{#1}\!\!\bra{#2}}
\newcommand{\proj}[1]{\ketbra{{#1}}{{#1}}}

\newcommand{\parahead}[1]{\noindent\textbf{{#1}}} 

\newcommand{\ComCla}[1]{\textup{\textbf{#1}}}
\newcommand{\sharpP}{\ComCla{\#P}}
\newcommand{\sharpBQP}{\ComCla{\#BQP}}
\newcommand{\sharpBPP}{\ComCla{\#BPP}}
\newcommand{\CeqP}{\ComCla{C\ensuremath{{}_=}P}}

\newcommand{\QMA}{\ComCla{QMA}}
\newcommand{\QCMA}{\ComCla{QCMA}}
\newcommand{\NQP}{\ComCla{NQP}}
\newcommand{\BQP}{\ComCla{BQP}}

\newcommand{\probl}[1]{\textup{\textsc{#1}}}

\newcommand{\la}{\lambda}

\newcommand{\aS}{\mathfrak{S}}

\newcommand{\sL}{\textsf{L}}
\newcommand{\sR}{\textsf{R}}

\newcommand{\tensor}{\ensuremath{\textstyle\bigotimes}}

\title{A remark on the quantum complexity of the Kronecker coefficients}
\author{Christian Ikenmeyer\thanks{University of Warwick, \texttt{christian.ikenmeyer@warwick.ac.uk}, supported by EPSRC grant  EP/W014882/1}~ and Sathyawageeswar Subramanian\thanks{University of Warwick, \texttt{sathya.subramanian@warwick.ac.uk}}}

\begin{document}
\raggedbottom
\maketitle

\begin{abstract}
We prove that the computation of the Kronecker coefficients of the symmetric group is contained in the complexity class $\sharpBQP$. This improves a recent result of Bravyi, Chowdhury, Gosset, Havlicek, and Zhu.
We use only the quantum computing tools that are used in their paper and additional classical representation theoretic insights. 
We also prove the analogous result for the plethysm coefficients.
\end{abstract}



\section{Introduction}
We study the Kronecker coefficients of the symmetric group, i.e., the representation theoretic multiplicities $k(\la,\mu,\nu)$ in the decomposition of the tensor product of Specht modules into irreducible representations:
\[
[\la]\otimes[\mu] = \bigoplus_\nu [\nu]^{\oplus k(\la,\mu,\nu)},
\]
where $\la,\mu,\nu\vdash n$ are partitions of $n$. Recently, Bravyi, Chowdhury, Gosset, Havlicek, and Zhu \cite{Bravyi2023quantum} proved that the computation of the product $k(\la,\mu,\nu)\cdot d(\la)\cdot d(\mu)\cdot d(\nu)$ is in the complexity class $\sharpBQP$, where $d(\la) = \dim[\la]$ is the number of standard tableaux of shape $\la$.
The input partitions are given in unary, as for example also in \cite{IMW17} or \cite{IF20}, and we also use this convention in this paper.
We refine the technique from \cite{Bravyi2023quantum} and tighten their result as follows.
\begin{theorem}
\label{thm:kronBQP}
The map $\textup{\textsc{Kron}}:(\la,\mu,\nu)\mapsto k(\la,\mu,\nu)$ is in $\sharpBQP$.
\end{theorem}
We also obtain an analogous result for the plethysm coefficient $a_\la(d,m)$, which is the representation theoretic multiplicity of the irreducible $\GL(V)$-representation $S_\la(V)$ in the nested symmetric power $S^d(S^m(V))$.
\begin{theorem}
\label{thm:plethBQP}
The map $\textup{\textsc{Pleth}}:(d,m,\la)\mapsto a_\la(d,m)$ is in $\sharpBQP$.
\end{theorem}

The question whether $\textsc{Pleth}\in\sharpP$ is a formalization of question 9 in Stanley's list of open positivity problems in algebraic combinatorics, \cite{Sta00}, see also \cite{IP22}.
The question whether $\textsc{Kron}\in\sharpP$ is a formalization of question 10 in \cite{Sta00}.
It has been resolved in several subcases, see \cite{remm:89,RW94,BO06,BMS15,T15,IMW17,Bla17,Liu17,BL18}.
The question has been explicitly mentioned again in \cite{Str22}, and in the surveys \cite{Pak22,Pan23a,Pan23b}.
For the plethysm problem, $\sharpP$ expressions have been found in some cases as well \cite{CDW12,MK16,PW19,IF20,dBPW21}. 
Both coefficients play an important role in geometric complexity theory, an approach towards computational complexity lower bounds via algebraic geometry and representation theory \cite{MS01,MS08,BLMW11,BI11,BI13,Kum15,BMS15,BHI17,IP17,BI18,BIP19,IK20,DIP20}.

\subsection{No rescaling after the computation}
\label{subsec:postdivision}
Let $\{0,1\}^*$ denote the set of finite length bit strings.
Recall that $\sharpP$ is the class of functions $f:\{0,1\}^*\to\IN$ such that there exists a nondeterministic polynomial time Turing machine $M$ such that for all inputs $w\in\{0,1\}^*$ the number of accepting paths of $M$ is exactly $f(w)$.

Let \probl{ScaledKron} be the function $(\la,\mu,\nu)\mapsto d(\la)\cdot d(\mu)\cdot d(\nu)\cdot k(\la,\mu,\nu)$.
\cite{Bravyi2023quantum} prove $\probl{ScaledKron}\in\sharpBQP$, see the definition in \S\ref{subsec:quantum},
but they do not prove that $\probl{Kron}\in\sharpBQP$.
Since $d(\la)$ can be computed in polynomial time via the hook length formula, it follows that
if $\probl{Kron}\in\sharpP$, then also $\probl{ScaledKron}\in\sharpP$, but the converse is not clear, because $\sharpP$ crucially does not allow any division after the counting.
In this section we list several examples of functions in $\sharpP$ that are multiples of simple functions (the function is always constant in these examples), but when divided by those, they are not expected to be in $\sharpP$, or at least it is wide open; see also the discussion in \cite{IP22} on this recent topic. We discuss these examples to explain our motivation to refine
$\probl{ScaledKron}\in\sharpBQP$ to $\probl{Kron}\in\sharpBQP$. These examples are not necessary for understanding our results.  

We first start with a naive example.
The standard example where a $\sharpP$ function \emph{can} be divided by a constant and still remain in $\sharpP$ is \probl{\#3Coloring}: Given a graph with at least one edge, determine the number of ways to color its vertices with 3 colors, so that no two adjacent vertices have the same color.
Since $\probl{3Coloring}\in\ComCla{NP}$, clearly $\probl{\#3Coloring}\in\sharpP$. Clearly $\forall G: \probl{\#3Coloring}(G)$ is divisible by 6, because the symmetric group $\aS_3$ acts on colorings by permuting the colors (here we used that $G$ has at least one edge).
We define $\frac{\probl{\#3Coloring}}{6}$
via $\frac{\probl{\#3Coloring}}{6}(G) = \frac{\probl{\#3Coloring}(G)}{6}$.
We have $\frac{\probl{\#3Coloring}}{6}\in\sharpP$, because for each 3-coloring $c$ it is easy to generate the set $\aS_3 c$ of all six 3-colorings, and decide if $c$ is the lexicographically smallest element in $\aS_3 c$ or not.
The Turing machine for 
$\frac{\probl{\#3Coloring}}{6}$
only accepts the 3-colorings $c$ that are lexicographically smallest in $\aS_3 c$. 

But for other problems it might not always be possible to easily generate the corresponding groupings from which one wants to select the lexicographically smallest element.

For an edge $e$ in a graph $G$ let $\probl{\#HCE}(G,e)$ denote the number of Hamiltonian cycles of $G$ that use $e$.
Clearly, $\probl{\#HCE} \in \sharpP$.
Smith's theorem \cite{Tut46} says that for any fixed edge $e$ in a cubic graph $G$ we have that $\probl{\#HCE}(G,e)$ is even.
However, finding a polynomial time self-inverse algorithm that maps one Hamiltonian cycle to its partner is a major open question.
In particular, it is an open question whether or not $\frac{\probl{\#HCE}}{2} \in \sharpP$.
The Price-Thomason lollipop algorithm \cite{Pri77,Tho78} for mapping a Hamiltonian cycle to its partner uses exponential time in the worst case \cite{Kra99,Cam01}.
The algorithm explores the so-called exchange graph $X$ of a graph $G$: A vertex in $X$ corresponds to a Hamiltonian path or cycle in $G$. 
\cite{CE99} list more problems of this kind, and they write:
\begin{quote}
Each proof consists of describing an ``exchange graph'' $X$, quite large compared to~$G$. [...] Each of these theorems is not so easy to prove without seeing the exchange graph.
\end{quote}

In the same vein,
in \cite{IP22} the problem
$\probl{\#COUNTALL-PPA${}_{\probl{Leaf}}$}$ 
is defined as follows:
Given two Boolean circuits $C_1$ and $C_2$, each with $n$ inputs and $n$ outputs,
let $N(v):=\{C_1(v),C_2(v)\}$.
Define a graph $G$ on the vertex set $\{0,1\}^n$
by defining that $v\in \{0,1\}^n$ and $w\in\{0,1\}^n$ are adjacent iff $v\neq w$ and $v \in N(w)$ and $w \in N(v)$.
Clearly, each vertex in $G$ is either isolated, or of degree 1, or of degree 2.
On input $(C_1,C_2)$, the function
$\probl{\#COUNTALL-PPA${}_{\probl{Leaf}}$}$
outputs the number of degree 1 vertices of~$G$.
This number is always even, but there exists an oracle $A$ such that $\frac{(\probl{\#COUNTALL-PPA${}_{\probl{Leaf}}$})^A}{2}\notin\sharpP^A$.

Another example is a promise problem using Karamata's inequality, which is also an even $\sharpP$ function, and when divided by 2 there is an oracle separation from $\sharpP$, see 
\cite[\S III.C]{IP22}.

In conclusion, in general it is not always the case that membership in a counting class is preserved under division by even very simple functions. Therefore it is valuable to refine counting results of scaled functions, such as $\probl{ScaledKron}\in\sharpBQP$
to
$\probl{Kron}\in\sharpBQP$.
If our refinement 
to
$\probl{Kron}\in\sharpBQP$ in Theorem~\ref{thm:kronBQP}
would \emph{not} have been easily possible, then it would have been justified to conjecture that $\probl{Kron}$ shares the fate of the problems listed in this subsection, which would have given a justification to conjecture $\probl{Kron}\notin\sharpP$.

\subsection{Remarks}
Even without our refinement, one can remark the following.
For any function $f \in \sharpP$, if the corresponding vanishing problem $[f = 0]$ is $\CeqP$-hard, then the polynomial hierarchy collapses to the second level, i.e., $\ComCla{PH} = \Sigma^\textsf{p}_2$, see Proposition~3.1.1 in \cite{IPP23}.
Let us assume for a moment that this approach works for $\probl{Kron}$,
i.e., assume that one could prove that $[\probl{Kron}=0]$ is $\CeqP$-hard,
i.e., $[\probl{Kron}>0]$ is \NQP-hard.
Since 
$[\probl{Kron}>0] \in \ComCla{QMA}$ by \cite{Bravyi2023quantum},
this would imply $\NQP \subseteq \QMA$.
See \cite{KMY03} for the relationship between $\NQP$ and~$\QMA$.
Therefore, a proof of $\probl{Kron}\notin\sharpP$ cannot just use the techniques in \cite{IPP23} without proving new results in quantum complexity theory.

Even proving $\probl{Kron}\in\sharpBPP$ is a major open problem, and also its quantum verifier analogue of counting $\QCMA$ witnesses: classical polynomial sized witnesses for $\probl{Kron}$ that can be verified by a $\BQP$ machine. Since these classes work with classical witnesses, from a combinatorial perspective they seem more useful than $\sharpBQP$, which is about ``counting $\QMA$ quantum witnesses''.
There exists a randomized classical oracle $A$ with $\QCMA^A \subsetneqq \QMA^A$, see \cite{NN23}.

The statement $\probl{Kron}\in\sharpBQP$ has been remarked without proof in \cite{HCW},
but as far as we know \cite{Bravyi2023quantum} is the first publication about the subject.

\section{Preliminaries}
\subsection{Representation theory}
We work over the complex numbers $\IC$.
Let $\aS_n$ denote the symmetric group on $n$ symbols.
A monotone nonincreasing finite sequence of nonnegative integers is called a \emph{partition}.
For a partition $\la$ we write $|\la| := \sum_{i\in\IN}\la_i$.
We also write $\la\vdash n$ to indicate that $\la$ is a partition with $|\la|=n$.
The length of a partition $\la$ is defined as $\ell(\la) := \max\{j \mid \la_j > 0\}$.
Let $\la^T$ denote the transpose partition of $\la$, which is defined via $\la^T_i := |\{j \mid \la_j\geq i\}|$.
The Young diagram of a partition $\la$ is defined as the set $\{(i,j)\mid 1 \leq i \leq \la^T_j, \ 1 \leq j \leq \la_i \}$, which is usually depicted as a top-left justified set of boxes with $\la_i$ boxes in row~$i$. For example, the Young diagram for $(5,3)$ is $\ytableausetup{smalltableaux}\ydiagram{5,3}$.
A Young diagram with $n$ boxes can be encoded for example as a sequence of $n-1$ bits, where the Young diagram is read row-wise from left to right and top to bottom, and a 1 represents a line break in the Young diagram, while a 0 represents the absence of a line break.
The Young diagram of $\la^T$ is obtained from the Young diagram of $\la$ by reflecting it about the main diagonal.
A representation of $\aS_n$ is a finite dimensional complex vector space $V$ with a group homomorphism $\varrho:\aS_n\to \GL(V)$.
A representation is called irreducible if it has no nontrivial subrepresentations.
The irreducible representations of $\aS_n$ are indexed by partitions $\la\vdash n$, see e.g.~\cite[\S7]{Ful97}.
For $\la\vdash n$ we denote by $[\la]$ the irreducible representation of $\aS_n$ of type~$\la$, called the Specht module. For example, $[(n)]$ is the trivial representation, and $[(1,1,\ldots,1)]$ is the sign representation. 
For an $\aS_n$-representation $V$ let $V^\la$ denote its $\la$-isotypic component, i.e., the direct sum of all its irreducible representations of isomorphism type $\la$.
Hence, with this notation, if $V$ is an $\aS_n$-representation, then the linear subspace of $\aS_n$-invariants is denoted by $V^{(n)}$.
If we have two commuting actions of $\aS_n$ on $V$, then we write $V^{\la,\ast}$ for the $\la$-isotypic component of the first action, and $V^{\ast,\la}$ for the $\la$-isotypic component of the second action, and $V^{\la,\mu} = V^{\la,\ast} \cap V^{\ast,\mu}$
for the $(\la,\mu)$-isotypic component for the action of the product group $\aS_n\times\aS_n$.
The group algebra $\IC[\aS_n]$ is the $n!$ dimensional vector space of formal linear combinations of permutations from~$\aS_n$.
On $\IC[\aS_n]$ we have two commuting actions, $\cdot_\sL$ and $\cdot_\sR$, as follows. Let $\pi \in\aS_n$ and let $\sum_{\sigma\in\aS_n}c_\sigma\,\sigma \in \IC[\aS_n]$. We define
\begin{itemize}
    \item $\pi \cdot_\sL (\sum_{\sigma\in\aS_n}c_\sigma\,\sigma) := \sum_{\sigma\in\aS_n}c_\sigma \,\pi\sigma$, and
    \item $\pi \cdot_\sR (\sum_{\sigma\in\aS_n}c_\sigma\,\sigma) := \sum_{\sigma\in\aS_n}c_\sigma\, \sigma\pi^{-1}$.
\end{itemize}

This gives an action of $\aS_n\times \aS_n$ on $\IC[\aS_n]$ and its decomposition into irreducibles is well known, e.g., \cite[Thm.~4.2.7]{GW09} (note that Specht modules are self-dual, because the characters of the symmetric group are real-valued):
\[
\IC[\aS_n]\simeq\bigoplus_{\la\vdash n} \ [\la]\otimes[\la].
\]
Hence, $\IC[\aS_n]^{\la,\ast} = \IC[\aS_n]^{\ast,\la} = \IC[\aS_n]^{\la,\la}$.

We embed $\aS_n\hookrightarrow\aS_n\times\aS_n\times\aS_n$ via $\pi\mapsto(\pi,\pi,\pi)$.
Given partitions $\la$, $\mu$, $\nu$,
using this embedding,
the Kronecker coefficient $k(\la,\mu,\nu)$ is defined as the dimension of the $\aS_n$-invariant space
\begin{equation}\label{eq:defkron}
k(\la,\mu,\nu) \ := \ \dim([\la]\otimes[\mu]\otimes[\nu])^{(n)}.
\end{equation}
For several equivalent definitions, see e.g.\ \cite[\S4.4]{Ike12}. \cite{Bravyi2023quantum} use the character theoretic characterization $k(\la,\mu,\nu) = \frac{1}{n!}\sum_{\pi\in\aS_n}\chi_\la(\pi)\chi_\mu(\pi)\chi_\nu(\pi)$, but in this paper we put more focus on invariant spaces instead.

For a partition $\la\vdash n$, the Young subgroup $\aS_\la \subseteq \aS_n$ is defined as $\aS_\la := \aS_{\la_1}\times \aS_{\la_2}\times \cdots \aS_{\ell(\la)}$ where the embedding is the usual one:
Recall that $\aS_n$ is the group of bijections from $\{1,\ldots,n\}$ to itself, whereas each factor
$\aS_{\la_i}$ is the group of bijections from $\{\la_1+\cdots+\la_{i-1}+1,\ldots,\la_1+\cdots+\la_{i}\}$ to itself.
For example, $\aS_{(2,2)} = \{\idd,(1\,2),(3\,4),(1\,2)(3\,4)\}$ is isomorphic to the Klein four-group.

A semistandard tableau of shape $\la$ is defined as a labeling of the boxes of the Young diagram of $\la$ with positive integers, such that the numbers strictly increase down each column, and are non-decreasing along each row. The content of a semistandard Young tableau $t$ is the vector that records at position $i$ the number of occurrences of label $i$ in~$t$.
The Kostka number $K_{\la,\mu}$ is defined as the number of semistandard tableaux of shape $\la$ and content $\mu$. For example, $K_{(3,1),(2,1,1)}=2$ because there are two semistandard tableaux of shape $(3,1)$ and content $(2,1,1)$: $\ytableausetup{smalltableaux}\ytableaushort{112,3}$ and~$\ytableausetup{smalltableaux}\ytableaushort{113,2}$.
Clearly, $K_{\la,\la}=1$, because there exists a unique semistandard tableau of shape $\la$ and content $\la$: The tableau with only entries $i$ in row $i$.
For a partition $\la$ and an $\aS_n$-representation $V$ let $V^{\aS_\la} \subseteq V$ denote the linear subspace of $\aS_\la$-invariants.
The following lemma is well known, but we give a proof based on Pieri's rule and invariants for the sake of completeness and because this is the main driving lemma of our Theorems~\ref{thm:kronBQP} and~\ref{thm:plethBQP}.

\begin{lemma}\label{lem:dimone}
For all $\la,\mu\vdash n$ we have $\dim([\la]^{\aS_\mu}) = K_{\la,\mu}$.
In particular, $\dim([\la]^{\aS_\la})=1$.
\end{lemma}
\begin{proof}
We will use the fact that the $\aS_\mu$-invariants can be iteratively obtained: $[\la]^{\aS_\mu} = ((\ldots([\la]^{(\mu_{\ell(\mu)})})^{(\mu_{\ell(\mu)-1})})^{\cdots})^{(\mu_{1})}$.
For $\mu \subseteq \la$ we write $\mu\sqsubset\la$ if
for all $i$ we have $\la^T_i-\mu^T_i \leq 1$.
Pictorially, the set of boxes $\la\setminus\mu$ contains at most 1 box in each column.
Let $i = \mu_{\ell(\mu)}$, $j=n-i$, $\aS_j\times\aS_i\hookrightarrow\aS_n$.
We use Pieri's rule, \cite[Cor.~3.5.13]{CSTS10}, which states how the $\aS_i$-invariant space $[\la]^{\aS_i}$ of the irreducible $\aS_n$-representation $[\la]$ decomposes into irreducible $\aS_j$-representations.
\[
[\la]^{\aS_i} = \bigoplus_{\substack{
\mu\vdash j, \, \mu\sqsubset\la
}}
[\mu]
\]
We take a Young diagram of shape $\la$ and 
mark the boxes in the difference $\la\setminus\mu$ with $\ell(\mu)$, and then repeat this process on each summand $[\mu]$
for
$i=\mu_{\ell(\mu)-1}$,
and $n$ decreased by~$i$. We continue in this way, and the process finishes when the boxes are marked with 1. All boxes are then marked and the result is a semistandard tableau of shape $\la$ and content $\mu$ by construction.
It is clear that for each semistandard tableau there is exactly one way to obtain it as the output of this process.
\end{proof}

\subsection{Quantum algorithms and \texorpdfstring{$\sharpBQP$}{\#BQP}}
\label{subsec:quantum}
We review some quantum computing preliminaries and notation, and refer to standard textbooks such as \cite{Nielsen2010} for details. Let $\H_n:={\bigotimes^n}\C^2$ be the Hilbert space of states of $n$-qubits spanned by the computational basis $\{\ket{x} \mid x\in\{0,1\}^n\}$, and denote $N:=\dim\H_n=2^n$.
Pure states of an $n$-qubit quantum system are projective vectors of unit $\ell_2$-norm in $\H$, where by projective we mean that for all vectors $\ket\psi\in\H$, we identify $\ee^{\ii\theta}\ket\psi$ with $\ket\psi$ for all $\theta\in[0,2\pi)$.

Let $\G\subset U(N)$ be a universal gateset and let $\qid_n\in U(N)$ be the $N\times N$ identity matrix. That is, for every $U\in U(N)$ and $\epsilon>0$ there exist $k=k(\epsilon,n)\in\IN$ and $g_1,\ldots,g_k\in\G$ such that $U_k:=g_k g_{k-1}\ldots g_1$ is $\epsilon$-close to $U$ in spectral norm. We call any product $C=g_s\ldots g_1$ with $\forall i:g_i\in\G$ a quantum circuit of size $s$ in the gateset $\G$. We will henceforth work with standard gatesets consisting of finitely many one and two qubit gates.

Any decomposition $\H=\H_1\oplus\H_2\oplus\ldots\oplus\H_k$ for $k\leq N$ defines a measurement on $\H$, but not all measurements can be performed efficiently. In this paper we will only be interested in Projection Valued Measures (PVM), defined by a collection of idempotent Hermitian endomorphisms $M=\{\Pi_1,\ldots,\Pi_k\}$ so that for each $i\in 1,\ldots, k$ we have $\Pi_i\H=\H_i$ (i.e.\ $\Pi_i$ projects onto $\H_i$), and $\Pi_i\Pi_j=\delta_{ij} \Pi_j$, i.e.\ the projections are pairwise orthogonal. All projectors that we will use in this paper are Hermitian,. Measuring a pure state $\ket\psi$ by $M$ produces as output a random variable $X$ taking values $1,\ldots,k$ with probability $\Pr(X=i)=\braket{\psi|\Pi_i|\psi}$, and a post-measurement state $\ket{\psi_i}:=\frac{\Pi_i\ket\psi}{\Pr(X=i)}$. Note that since $\Pi_i$ is a projector, $\braket{\psi|\Pi_i|\psi}=\Norm{\Pi_i\ket{\psi}}^2$. A computational basis measurement of $\H_n$ is the PVM defined by $\{\Pi_x=\proj{x}\mid x\in\{0,1\}^n\}$. A single qubit computational basis measurement is called accepting if the outcome is one, and rejecting otherwise.

We now formally define the two main complexity classes that the results of this article are concerned with.

\paragraph{QMA:}
A language $L\in\{0,1\}^*$ is in the class $\QMA$
if there is a polynomial time Turing machine that on input $w$, $|w|=n$, outputs a quantum circuit $C_w$ acting on $m+a$ qubits with $m,a=\poly n$,
such that for every $n\in\IN$ and every $w\in\{0,1\}^n$ we have that
\begin{itemize}
    \item $C_{w}$ takes as input $\ket{\psi}\in\H_m$ and $a$ many ancillary qubits in the initial state $\ket{0^a}\in\H_a$ and outputs the single bit outcome of measuring the first qubit in the computational basis $\{\Pi_0,\Pi_1\}$, and
    \item (completeness) if $w\in L$, then $\exists \ket{\psi_w}\in\H_m$  such that $\Norm{\Pi_1 C_w \ket{\psi_w}\ket{0^a}}^2\geq 2/3$, and
    \item (soundness) if $w\notin L$, then $\forall \ket\psi\in\H_m$, $\Norm{\Pi_1 C_w \ket{\psi}\ket{0^a}}^2\leq 1/3$,
\end{itemize}
where for convenience we have used $\Pi_1$ to implicitly mean $\Pi_1\otimes\qid_{m+a-1}$.

If in the second and third bullet point in the definition of $\QMA$ we restrict the witness states $\ket{\psi_w}$ and $\ket\psi$ to be classical, i.e.\ basis vectors $\{\ket{y}\mid y\in\{0,1\}^m\}$ in the computational basis, then we get the class $\QCMA$.

\paragraph{\#BQP:}
A function $f:\{0,1\}^*\to\IN$ is in $\sharpBQP$
if there is a polynomial time Turing machine that on input $w$, $|w|=n$, outputs a quantum circuit $C_w$ acting on $m+a$ qubits with $m,a=\poly n$  such that there exists a decomposition $\H_m = \mathcal A_w \oplus \mathcal R_w$ of $\H_m$ into linear subspaces
of so-called \emph{accepting and rejecting witnesses} such that
\begin{samepage}
\begin{itemize}
 \item $\Norm{\Pi_1 C_w\ket\psi\ket{0^a}}^2 \geq 2/3$ for all $|\psi\rangle \in \mathcal A_w$, and
 \item $\Norm{\Pi_1 C_w\ket\psi\ket{0^a}}^2 \leq 1/3$ for all $|\psi\rangle \in \mathcal R_w$, and
 \item $f(w)=\dim\mathcal A_w$.
\end{itemize}

See \cite{SZ09,BFS11} for a more detailed discussion of this complexity class.
\end{samepage}
The circuit $C_w$ can be considered the quantum verifier in a $\QMA$ protocol, and $f(w)$ is the dimension of the subspace of witnesses accepted with probability at least $2/3$.

The acceptance and rejection conditions above are described by an operator $E_w:=(\qid_m\otimes\bra{0^a})C_w^{\dagger}(\Pi_1\otimes\qid_{m+a-1})C_w(\qid_m\otimes\ket{0^a})$ on $\H_m$, which is a positive operator because $\qid_m\otimes\ket{0^a}$ is an isometric embedding $\H_m\hookrightarrow\H_{m+a}$, $C_w$ is a unitary on $\H_{m+a}$, and $\Pi_1\otimes\qid_{m+a-1}$ is a projector on $\H_{m+a}$. However this operator might not in general define a \emph{PVM} on $\H_m$; but for our purpose we will construct it to be one, as in \cite{Bravyi2023quantum}. Consequently, when $E_w$ is a projector, $\mathcal A_w=\im E_{w}$, $\dim\mathcal A_w=\tr E_{w}$. In this case, $C_w$ accepts witness states in $\mathcal A_w$ with certainty, and rejects states in $\mathcal R_w$ with certainty, corresponding to the class $\QMA_1^1$ with perfect soundness and completeness. It is known that there is a quantum oracle relative to which $\QMA_1\neq \QMA$ for the case of perfect completeness \cite{aaronson2008perfect}. Interestingly, it is known in contrast that $\QCMA_1=\QCMA$, under any gateset in which the Hadamard and reversible classical operations can be implemented exactly \cite{Jordan12qcma}.

We say that a PVM $\{\Pi_i\}$ on $n$-qubits can be efficiently implemented if we can construct a polysize unitary circuit $U_i$ acting on $n+a$ many qubits where $a=\poly n$ such that on an arbitrary input state $\ket\psi\ket{0^a}\in\bigotimes^{n+a}\C^2$ the circuit sets a designated qubit to $1$ with probability $p_i:=\braket{\psi|\Pi_i|\psi}$. To be more precise, the circuit implements a unitary map $\ket\psi\ket{0^a}\ket{0}\mapsto \sqrt{p_i}\ket\psi\ket{0^a}\ket{1} + \ee^{\ii\theta}\sqrt{1-p_i}\ket\psi\ket{0^a}\ket{0}$ where $\theta\in[0,2\pi)$ can be an arbitrary phase. 

Note that using $n^2$ bits one can encode a permutation $\pi\in\aS_n$ as its permutation matrix, or more specifically, as its $n^2$ entries when reading row-wise from left to right. Hence,
for our \probl{Kron} problem, $\IC[\aS_n]$ can be embedded into $\H_m$ with $m = n^2$
such that any function
$f \in \IC[\aS_n]$ with $||f||_2 = 1$ can be written in the computational basis as
\[
    \ket{f} = \sum_{\sigma\in \aS_n}f(\sigma)\ket{\sigma}.
\]
Here, $\ket\sigma$ is shorthand for $\ket{\mathrm{enc}(\sigma)}$, with $\mathrm{enc}(\sigma)$ being the $m$-bit string that encodes $\sigma$. Similarly, the input register $\tensor^3\H_m$ of the quantum algorithm can be thought of as $\IC[\aS_n]\otimes\IC[\aS_n]\otimes\IC[\aS_n]$.
The basis vectors in the Fourier basis $\IC[\aS_n]\simeq [\la]\otimes[\la]$
are indexed by triples $(\la,i,j)$,
where $\la\vdash n$ and $i$ and $j$ are standard tableaux of shape~$\la$.
The quantum Fourier transform
maps a basis vector in the Fourier basis to a basis vector in the computational basis, with a suitable computational basis encoding for the triple $(\la,i,j)$.

We write
$\rightrightarrows$
to denote a Hermitian idempotent linear map, i.e., a Hermitian projector.
Note that if two Hermitian projectors commute, then their composition is again a Hermitian projector.
 \cite{Bravyi2023quantum} explain how the following two Hermitian projectors can be efficiently implemented: 
\begin{itemize}
    \item Weak Fourier sampling $\Pi^\la : \IC[\aS_n] \rightrightarrows \IC[\aS_n]^{\la,\la} \ = \
    \IC[\aS_n]^{\la,\ast} \ = \ \IC[\aS_n]^{\ast,\la} \ \simeq \ [\la]\otimes[\la]$. The set of pairwise orthogonal projectors $M=\{\Pi^{\la}\mid\la\vdash n\}$ is a PVM on $\H_m$ and can be implemented using the Weak Fourier sampling circuit \cite[Fig.\ 1a]{Bravyi2023quantum} that acts on $\H_m$ and an ancillary register $\H_{n-1}$ into which the partition label can be copied. A measurement is performed on the ancilla and the circuit rejects unless the outcome is $\la\vdash n$. If the procedure on input $\ket\psi$ accepts, then the post-measurement state of the $\H_m$ register is proportional to $\Pi^{\la}\ket\psi$.
    \item
    For any $\aS_k$-representation $(V,\varrho)$ we have the 
    generalized phase estimation procedure $P^{\aS_k} : V \rightrightarrows V^{(k)}$.
    Here we require that
    $V$ has a basis $B = \{v_i\}$ labeled by polynomially sized labels $\textsf{label}(v_i)$, and that
    $\forall \pi \in \aS_k : \pi v_i \in B$,
    and
    there is a uniform classical circuit that on input $(\pi,\textsf{label}(v))$, where $\pi\in\aS_k$,
    outputs $\textsf{label}(\pi(v))$.
    This will be obvious in all our cases:
    The product of transposition matrices can be computed using a small circuit, and the inverse of a transposition matrix (needed for the right action) is just its transpose. 
    
    The projector $P^{\aS_k}$ can be implemented using the generalized phase estimation circuit \cite[Fig.\ 1b]{Bravyi2023quantum} that acts on $\H_m\otimes\H_v$ where the first register is interpreted as an ancillary register, and the second one encodes $V$. At the end it implements $\Pi^{(n)}$ on the ancilla: The procedure rejects unless the measurement outcome is the trivial representation $(n)$. If the procedure on input $\ket\psi$ accepts, then the post-measurement state of the $\H_v$ register is proportional to $P^{\aS_k}\ket\psi$.
\end{itemize}


Non-adaptive intermediate measurements in a circuit can be deferred to the end without affecting the probability distribution of the output. The multiqubit measurements required in the above two cases can instead be replaced by a postprocessing circuit that checks if the registers to be measured are in the required state, and sets a single ancilla to $\ket1$ if yes and $\ket0$ otherwise.

We note that a subtlety arises in our use of $P^{\aS_k}$: The action will be an action of different subgroups $\aS_k\subseteq\aS_n$ 
on $\IC[\aS_n]$, 
sometimes from the left, and sometimes from the right. Both actions are readily implementable using small circuits.

Even though more general constructions are possible (see e.g.\ \cite{MRR06}), we only use these two constructions as building blocks to achieve our goal, but we make a refined use of generalized phase estimation.

\section{Proof of Theorem~\ref{thm:kronBQP}: Three-step construction}
\label{sec:construction}

To indicate more clearly which tensor position is which (in particular when reordering the positions), we write
$\IC[\aS_n]_1\otimes\IC[\aS_n]_2\otimes\IC[\aS_n]_3$
instead of just $\IC[\aS_n]\otimes\IC[\aS_n]\otimes\IC[\aS_n]$.
We will construct an idempotent endomorphism of 
$\IC[\aS_n]_1\otimes\IC[\aS_n]_2\otimes\IC[\aS_n]_3$
whose image has dimension $k(\la,\mu,\nu)$, and it will be a composition of commuting Hermitian projectors, each of which is either implementable via weak Fourier sampling or via generalized phase estimation. This then finishes the proof of Theorem~\ref{thm:kronBQP}.
The overview of the maps is depicted in Figure~\ref{fig:maps}.
\begin{figure}
    \centering
    \begin{eqnarray*}
\IC[\aS_n]_1\otimes\IC[\aS_n]_2\otimes\IC[\aS_n]_3
&\simeq&
\bigoplus_{\substack{\la\vdash n\\\mu\vdash n\\\nu\vdash n}} 
\Big(
[\la]_{\sL_1}\otimes[\la]_{\sR_1}\otimes[\mu]_{\sL_2}\otimes[\mu]_{\sR_2} \otimes [\nu]_{\sL_3}\otimes[\nu]_{\sR_3}
\Big)
\\
&\stackrel{\Pi_{1}^\la \otimes \Pi_{2}^\mu \otimes \Pi_{3}^\nu}{\rightrightarrows} &
[\la]_{\sL_1}\otimes[\la]_{\sR_1}\otimes[\mu]_{\sL_2}\otimes[\mu]_{\sR_2} \otimes [\nu]_{\sL_3}\otimes[\nu]_{\sR_3}
\\
&\stackrel{P_{\sL}^{\aS_n}}{\rightrightarrows} & 
([\la]_{\sL_1}
\otimes[\mu]_{\sL_2}
\otimes [\nu]_{\sL_3})^{\aS_n}
\otimes[\la]_{\sR_1}
\otimes[\mu]_{\sR_2}
\otimes[\nu]_{\sR_3}
\\
&\stackrel{P_{\sR_1}^{\aS_\la} \otimes P_{\sR_2}^{\aS_\mu} \otimes P_{\sR_3}^{\aS_\nu}}{\rightrightarrows}&
\underbrace{([\la]_{\sL_1} \otimes [\mu]_{\sL_2}
\otimes[\nu]_{\sL_3})^{\aS_n}}_{\dim=k(\la,\mu,\nu)}
\otimes\underbrace{[\la]_{\sR_1}^{\aS_\la}\otimes[\mu]_{\sR_2}^{\aS_\mu}  \otimes[\nu]_{\sR_3}^{\aS_\nu}}_{\dim=1}
\end{eqnarray*}
    \caption{The commuting Hermitian projectors that project to a space of dimension $k(\la,\mu,\nu)$.}
    \label{fig:maps}
\end{figure}
On
$\IC[\aS_n]_1\otimes\IC[\aS_n]_2\otimes\IC[\aS_n]_3$
we have an action of $(\aS_n)^6$: on each of the three factors we have an action of $(\aS_n)^2$, and the three actions of $(\aS_n)^2$ commute.

\subsection{First step: Projection to the isotypic component}

The following three Hermitian projectors commute and hence can be composed into one Hermitian projector:
\begin{eqnarray*}
\Pi_{1}^\la \otimes \Id_2 \otimes \Id_3  &:&  \IC[\aS_n]_1\otimes\IC[\aS_n]_2\otimes\IC[\aS_n]_3 \ \rightrightarrows \  
\IC[\aS_n]_1^{\la,\la}\otimes\IC[\aS_n]_2\otimes\IC[\aS_n]_3,
\\
\Id_1 \otimes \Pi_{2}^\mu \otimes \Id_3  &:&  \IC[\aS_n]_1\otimes\IC[\aS_n]_2\otimes\IC[\aS_n]_3 \ \rightrightarrows \  \IC[\aS_n]_1\otimes\IC[\aS_n]_2^{\mu,\mu}\otimes\IC[\aS_n]_3,
\\
\Id_1 \otimes \Id_2 \otimes \Pi_{3}^\mu  &:&  \IC[\aS_n]_1\otimes\IC[\aS_n]_2\otimes\IC[\aS_n]_3 \ \rightrightarrows \  \IC[\aS_n]_1\otimes\IC[\aS_n]_2\otimes\IC[\aS_n]_3^{\nu,\nu}.
\end{eqnarray*}

The composition of these three idempotents is called $\Pi_{1}^\la\otimes \Pi_{2}^\mu \otimes \Pi_{3}^\nu$.
It projects $\IC[\aS_n]_1\otimes\IC[\aS_n]_2\otimes\IC[\aS_n]_3$ to
the $(\la,\la,\mu,\mu,\nu,\nu)$-isotypic component $(\IC[\aS_n]_1\otimes\IC[\aS_n]_2\otimes\IC[\aS_n]_3)^{\la,\la,\mu,\mu,\nu,\nu}$.
Recall that
\begin{eqnarray*}
(\IC[\aS_n]_1\otimes\IC[\aS_n]_2\otimes\IC[\aS_n]_3)^{\la,\la,\mu,\mu,\nu,\nu}
&=&
(\IC[\aS_n]_1\otimes\IC[\aS_n]_2\otimes\IC[\aS_n]_3)^{\la,\ast,\mu,\ast,\nu,\ast}
\\
&=&
(\IC[\aS_n]_1\otimes\IC[\aS_n]_2\otimes\IC[\aS_n]_3)^{\ast,\la,\ast,\mu,\ast,\nu}
\end{eqnarray*}


Note that $\Pi_{1}^\la\otimes \Pi_{2}^\mu \otimes \Pi_{3}^\nu$ can be implemented by a quantum circuit using weak Fourier sampling three times \cite{Bravyi2023quantum}.

\subsection{Second step: Projection to the global invariants from the left}

We embed $\aS_n\hookrightarrow(\aS_n)^6$ via $\pi\mapsto(\pi,\idd,\pi,\idd,\pi,\idd)$.
We define $P_\sL^{\aS_n}$ to be the projector to the $\aS_n$-invariant space (called $Q$ in \cite{Bravyi2023quantum}).

Note that $P_\sL^{\aS_n}$ and $\Pi_{1}^\la\otimes \Pi_{2}^\mu \otimes \Pi_{3}^\nu$ commute, because $\Pi_{1}^\la\otimes \Pi_{2}^\mu \otimes \Pi_{3}^\nu$ can be defined in terms of the right actions, i.e., the unused three copies of $\aS_n$:
$(\Pi_{1}^\la\otimes \Pi_{2}^\mu \otimes \Pi_{3}^\nu)(V) = (V)^{\ast,\la,\ast,\mu,\ast,\nu}$.
We conclude that 
$P_\sL^{\aS_n} \circ \Pi_{1}^\la\otimes \Pi_{2}^\mu \otimes \Pi_{3}^\nu$ is idempotent.

\subsection{Third step: Refinement from the right}
For an $\aS_n$-representation $V$ and a partition $\la\vdash n$
we define $P^{\aS_\la}:V\to V^{\aS_\la}$ to be the projection to the Young subgroup invariant space of the action of $\aS_\la\subseteq\aS_n$.
Note that the elements from different factors $\aS_{\la_i}$ of $\aS_\la$ commute. Hence $P^{\aS_\la}$ can be written as the composition of commuting projectors to invariant spaces:
\[
P^{\aS_\la} \ = \ P^{\aS_{\la_1}} \circ P^{\aS_{\la_2}} \circ \cdots
\]
Each factor can be implemented using generalized phase estimation.
On $\IC[\aS_n]_1\otimes\IC[\aS_n]_2\otimes\IC[\aS_n]_3$
we have three of these projectors:
$P_{\sR_1}^{\aS_\la}$, $P_{\sR_2}^{\aS_\la}$,
and $P_{\sR_3}^{\aS_\la}$, using copies $2$, $4$, and $6$ of $(\aS_n)^6$.
Since they use different copies of $\aS_n$, these three projectors commute, and we call their composition 
$P_{\sR_1}^{\aS_\la} \otimes P_{\sR_2}^{\aS_\la} \otimes P_{\sR_3}^{\aS_\la}$.

Clearly $P_{\sR_1}^{\aS_\la} \otimes P_{\sR_2}^{\aS_\la} \otimes P_{\sR_3}^{\aS_\la}$ commutes with $P_\sL^{\aS_n}$, because they use disjoint copies of $\aS_n$.
A crucial insight is that the projector $\Pi_{1}^\la\otimes \Pi_{2}^\mu \otimes \Pi_{3}^\nu$ can also be expressed in terms of the \emph{left} actions:
$(\Pi_{1}^\la\otimes \Pi_{2}^\mu \otimes \Pi_{3}^\nu)(V) = V^{\la,\ast,\mu,\ast,\nu,\ast}$.
Hence, $\Pi_{1}^\la\otimes \Pi_{2}^\mu \otimes \Pi_{3}^\nu$ and $P_{\sR_1}^{\aS_\la} \otimes P_{\sR_2}^{\aS_\la} \otimes P_{\sR_3}^{\aS_\la}$ commute.
We conclude that 
$P_{\sR_1}^{\aS_\la} \otimes P_{\sR_2}^{\aS_\la} \otimes P_{\sR_3}^{\aS_\la} \circ P_\sL^{\aS_n} \circ \Pi_{1}^\la\otimes \Pi_{2}^\mu \otimes \Pi_{3}^\nu$ is a Hermitian projector.

\subsection{Kronecker dimension}
To prove Theorem~\ref{thm:kronBQP}
it remains to prove that $\dim\left(\im P_{\sR_1}^{\aS_\la} \otimes P_{\sR_2}^{\aS_\la} \otimes P_{\sR_3}^{\aS_\la} \circ P_\sL^{\aS_n} \circ \Pi_{1}^\la\otimes \Pi_{2}^\mu \otimes \Pi_{3}^\nu\right) = k(\la,\mu,\nu)$.
This proof is illustrated in Figure~\ref{fig:maps}.
It is instructive to think about the spaces in their Fourier basis, i.e.,
$\IC[\aS_n]_1\otimes\IC[\aS_n]_2\otimes\IC[\aS_n]_3
\simeq
\bigoplus_{{\la\vdash n,\,\mu\vdash n,\,\nu\vdash n}} 
\Big(
[\la]_{\sL_1}\otimes[\la]_{\sR_1}\otimes[\mu]_{\sL_2}\otimes[\mu]_{\sR_2} \otimes [\nu]_{\sL_3}\otimes[\nu]_{\sR_3}
\Big)
$.
The first projector $\Pi_{1}^\la \otimes \Pi_{2}^\mu \otimes \Pi_{3}^\nu$ maps onto the $(\la,\la,\mu,\mu,\nu,\nu)$-isotypic component, which is
$[\la]_{\sL_1}\otimes[\la]_{\sR_1}\otimes[\mu]_{\sL_2}\otimes[\mu]_{\sR_2} \otimes [\nu]_{\sL_3}\otimes[\nu]_{\sR_3}$.
The second projector $P_{\sL}^{\aS_n}$ maps
this space
to the left invariant space
$(([\la]_{\sL_1}
\otimes[\mu]_{\sL_2}
\otimes [\nu]_{\sL_3})^{\aS_n}
\otimes[\la]_{\sR_1}
\otimes[\mu]_{\sR_2}
\otimes[\nu]_{\sR_3}$.
The third projector $P_{\sR_1}^{\aS_\la} \otimes P_{\sR_2}^{\aS_\mu} \otimes P_{\sR_3}^{\aS_\nu}$
maps each right Specht module to a one dimensional space, according to Lemma~\ref{lem:dimone}:
$([\la]_{\sL_1} \otimes [\mu]_{\sL_2}
\otimes[\nu]_{\sL_3})^{\aS_n}
\otimes([\la]_{\sR_1})^{\aS_\la}\otimes([\mu]_{\sR_2})^{\aS_\mu}  \otimes([\nu]_{\sR_3})^{\aS_\nu}$.
By \eqref{eq:defkron}, the dimension of this last vector space is $k(\la,\mu,\nu)$.

\section{Proof of Theorem~\ref{thm:plethBQP}: An analogous construction}
We embed $\varphi:\aS_d \hookrightarrow \aS_{md}$ via
$\varphi(\sigma)(ad-d+b) := \sigma(a)d-d+b$ for all $1\leq a \leq d$, $1 \leq b \leq m$.
Consider the partition $(m^d) := (m,m,\ldots,m)$ of $md$, and 
consider the Young subgroup $\aS_{(m^d)} \simeq \aS_m \times \aS_m \times \cdots \times \aS_m \subset \aS_{md}$.
The wreath product $\aS_m\wr\aS_d \subset \aS_{md}$ is the subgroup generated by $\aS_{(m^d)}$ and the image of $\varphi$.
For an $\aS_{md}$-representation $V$ the invariant space $V^{\aS_{(m^d)}}$ carries the action of $\aS_{d}$ via $\varphi$, and its invariant space $(V^{\aS_{(m^d)}})^{\aS_d}$ is exactly the invariant space $V^{\aS_d\wr\aS_m}$.

For a $\GL(V)$-representation $W$, the symmetric power $\Sym^d W$ is again a $\GL(V)$-representation, isomorphic to the invariant space $(\tensor^d W)^{\aS_d}$.
The plethysm coefficient $a_\la(d,m)$ is the multiplicity of the irreducible $\GL(V)$-representation of type $\la$ in $\Sym^d(\Sym^m V)$.
This can also be expressed in terms of the symmetric group:
$a_\la(d,m)$ is the dimension of the wreath product invariant space $[\la]^{\aS_m\wr\aS_d}$.
This can readily be seen from Schur-Weyl duality \cite[eq.~(9.1)]{GW09}
\begin{eqnarray*}
\tensor^{dm}V
&\simeq&
\bigoplus_{\la} S_\la(V) \otimes [\la]
\\
\Longrightarrow \qquad \Sym^d(\Sym^m V)
\simeq
((\tensor^{dm}V)^{\aS_{(m^d)}})^{\aS_d}
\simeq
(\tensor^{dm}V)^{\aS_m\wr\aS_d}
&\simeq&
\bigoplus_{\la} S_\la(V) \otimes [\la]^{\aS_m\wr\aS_d}
\end{eqnarray*}

The proof of Theorem~\ref{thm:plethBQP} is analogous to the proof of  Theorem~\ref{thm:kronBQP} with some minor adjustments.
The commuting projectors are as follows:
    \begin{eqnarray*}
\IC[\aS_n]
&\simeq&
\bigoplus_{\substack{\la\vdash n}} 
\ 
[\la]_{\sL}\otimes[\la]_{\sR}
\\
&\stackrel{\Pi^\la}{\rightrightarrows} &
[\la]_{\sL}\otimes[\la]_{\sR}
\\
&\stackrel{P_{\sL}^{\aS_m\wr\aS_d}}{\rightrightarrows} & 
[\la]_{\sL}^{\aS_m\wr\aS_d}
\otimes[\nu]_{\sR}
\\
&\stackrel{P_{\sR}^{\aS_\la}}{\rightrightarrows}&
[\la]_{\sL}^{\aS_m\wr\aS_d}
\otimes[\nu]_{\sR}^{\aS_\la}
\end{eqnarray*}
The projectors commute for the same reason as in the proof of Theorem~\ref{thm:kronBQP}, and
\[
\dim\im(P_\sR^{\aS_\la}\circ P_\la^{\aS_m\wr\aS_d}\circ\Pi^\la) = \dim[\la]^{\aS_m\wr\aS_d} = a_{\la}(d,m).
\]
$P_{\sL}^{\aS_m\wr\aS_d}$ can be implemented via two commuting projectors:
$P_{\sL}^{\aS_m\wr\aS_d} = P_{\sL}^{\aS_{(d)}}\circ P_{\sL}^{\aS_{(m^d)}}$,
where for $P_{\sL}^{\aS_{(d)}}$ we use the action via $\varphi$.
This proves Theorem~\ref{thm:plethBQP}.

\parahead{Acknowledgements: }We would like to thank Michael Walter for helpful comments on our draft.


\printbibliography

\end{document}